\documentclass[english]{cccconf}
\usepackage{amsmath,amsfonts}
\usepackage{amsmath}
\usepackage{amsthm}
\usepackage{algorithmic}
\usepackage{array}
\usepackage[caption=false,font=normalsize,labelfont=sf,textfont=sf]{subfig}
\usepackage{textcomp}
\usepackage{stfloats}
\usepackage{url}
\usepackage{verbatim}
\usepackage{graphicx}
\usepackage[comma,numbers,square,sort&compress]{natbib}
\usepackage{epstopdf}
\newtheoremstyle{myproofstyle}
{\topsep}   % ABOVESPACE
{\topsep}   % BELOWSPACE
{\itshape}  % BODYFONT
{0pt}       % INDENT (empty value is the same as 0pt)
{\bfseries} % HEADFONT
{.}         % HEADPUNCT
{5pt plus 1pt minus 1pt} % HEADSPACE
{}          % CUSTOM-HEAD-SPEC

\theoremstyle{myproofstyle}

\begin{document}

\title{Online Optimal Parameter Compensation method of High-dimensional PID Controller for Robust stability}

\author{Sheng Zimao\aref{amss},
        Hong'an Yang\aref{amss}
        }

\affiliation[amss]{School of Mechanical Engineering, Northwestern Polytechnical University,Shaanxi, 710072, P.~R.~China
        \email{hpshengzimao@mail.nwpu.edu.cn}}

\maketitle

\begin{abstract}

Classical PID control is widely applied in an engineering system, with parameter regulation relying on a method like Trial - Error Tuning or the Ziegler - Nichols rule, mainly for a Single - Input Single - Output (SISO) system. However, the industrial nonlinear Multiple - Input Multiple - Output (MIMO) system demands a high - robustness PID controller due to strong state coupling, external disturbances, and faults. Existing research on PID parameter regulation for a nonlinear uncertain MIMO system has a significant drawback: it's limited to a specific system type, the control mechanism for a MIMO nonlinear system under disturbances is unclear, the MIMO PID controller over - relies on decoupled control, and lacks dynamic parameter compensation. This paper theoretically analyzes a high - dimensional PID controller for a disturbed nonlinear MIMO system, providing a condition for online dynamic parameter regulation to ensure robust stability. By transforming the parameter regulation into a two - stage minimum eigenvalue problem (EVP) solvable via the interior point method, it enables efficient online tuning. The experiment proves that the designed dynamic compensation algorithm can achieve online robust stability of system errors considering multi - channel input coupling, addressing the key limitation in the field. 
 
\end{abstract}

\keywords{Disturbed nonlinear MIMO system, dynamic compensation, robust stability, minimum eigenvalue problem}

% Please remove or comment out the following line if the footnote is not necessary
\footnotetext{This work is supported in part by the National Natural Science Foundation of China under Grant No.2017KA050099.}

\section{Introduction}
Classical proportional - integral - derivative (PID) control is the most widely used feedback - based control algorithm, applied in over 95$\%$ of control loops in engineering control systems \cite {bib:Åström}, which holds an irreplaceable role \cite {bib:Samad} because of its strong robustness against uncertainties.

Currently, the parameter adjustment of PID controllers predominantly depends on methods like the Trial and Error Tuning approach\cite{bib:O'dwyer} and the Ziegler - Nichols rules\cite{nie2022unifying}. In the context of the robust PID control parameter tuning problem, the frequency - domain internal model control (IMC) \cite{bib:Vilanova}, linear matrix inequality (LMI) \cite{bib:Ge}, etc\cite{bib:Killingsworth}\cite{bib:Verma}., are most frequently applied in the effective stabilization process of control error. The majority of the aforementioned methods are designed for perturbed linear or affine nonlinear Single - Input Single - Output system (SISO). However, most of the current industrial systems exhibit characteristics such as strong coupling of state variables, strong external disturbances, undesired sensor faults and actuator faults. These features pose higher requirements for the design of a high-robustness PID controller with MIMO, nonlinearity, and anti-disturbance capabilities. 

In an effort to investigate the PID parameter regulation problem of nonlinear uncertain Multiple - Input Multiple - Output system (MIMO) system, a substantial number of theoretical outcomes have been obtained\cite{bib:Zhao}\cite{song2018control}. Jinke\cite{bib:Zhang} offers an explicit method to construct PID parameters using the upper bounds of derivatives of unknown nonlinear drift and diffusion terms in a stochastic mechanic system with whitenoise and linear inputs. Both Cheng\cite{zhao2017pid} and Guo\cite{bib:Zhao} theoretically derived sufficient conditions for ensuring the effective error stabilization of PID controllers. Specifically, the controller parameters are constrained within a range associated with the Lipschitz constant of the map of the nonlinear system. Nevertheless, the acquisition of such a Lipschitz constant is typically challenging, thus impeding its practical implementation. Cheng\cite{bib:Zhao2} devised a PID controller by factoring in the high - order differential terms of the error. From a theoretical perspective, a set of parameter tuning rules were formulated for SISO affine - nonlinear uncertain systems, which can guarantee that the error converges exponentially to zero. To secure more optimal outcomes in practical deployment, numerous scholars have embarked on painstaking endeavors. Among them, a representative work is that of Song \cite{bib:Son}. For a nonlinear SISO system under the consideration of the first - order plus time delay (FOPTD) model, Song aimed to minimize three performance indices and presented a sufficient condition in the frequency domain for the parameters of a nonlinear PID controller, which can ensure the absolute stability of the system error. Most existing studies exhibit the following issues. First, the majority of current PID controllers are predicated on specific types of systems with disturbances, such as linear systems, affine nonlinear systems, and SISO systems. Designing controllers applicable to MIMO nonlinear systems in a general form and with more general disturbances remains a challenge. Second, the control mechanism by which a PID controller can control MIMO nonlinear systems with disturbances remains unclear, impeding the implementation of an online regulation method that can effectively guarantee robust stability. Additionally, most existing MIMO controllers rely on decoupled control, making it difficult to account for the interactions among multiple input channels in a coupled manner to achieve optimal parameter configuration. Moreover, during the control process, the gains of the PID controller remain constant, rendering it arduous to achieve automatic dynamic compensation of the controller coefficients in response to sudden environmental changes and disturbances. To address these problems, we theoretically analyzed the control mechanism of high - dimensional PID controller applied to disturbed MIMO system and provided sufficient conditions for online dynamic regulation of controller parameters to ensure the robust stability of the system. To minimize the invariant set region for the robust stability of the system, we transformed the controller parameter regulation problem into a two - stage minimum eigenvalue problem (EVP), which can be solved online using the interior point method. Experimental results demonstrate that the dynamic compensation algorithm we designed can effectively achieve online robust stability of system errors while considering the coupling of multi - channel input commands. 

The structure of this paper is as follows. Section 1 is the introduction, which presents the research progress at home and abroad regarding online regulation of PID controller parameters for MIMO systems and the contributions of this paper. Section 2 introduces basic definitions and lemmas, provides the velocity form of high - dimensional PID controllers applied to general disturbed MIMO systems, and clarifies the fundamental framework of the research problem. Section 3 details the design strategy of high - dimensional PID controllers with dynamic parameter regulation. Section 4 designs specific experiments to validate the proposed algorithm. Section 5 summarizes the above content.

\section{Problem Formulation}
\label{sec:formulation}
\subsection{Notations}
Denote $\mathbb{R}^n$ as the n-dimensional Euclidean space, $\mathbb{R}^{m\times n}$ as the space of $m\times n$ real matrices, $||x||_2 = \sqrt{x^Tx}$ as the Euclidean 2-norm of a vector x. The norm of a matrix $P\in \mathbb{R}^{m\times n}$ is defined by $||P||_2=\sup_{x\in R^n,||x||_2=1}||Px||_2=\sqrt{\lambda_{\max}(P^TP)}$, for given matrix set $\mathbb{P}$, its 2-norm is defined as $||\mathbb{P}||_2=\arg\sup_{P\in\mathbb{P}}||P||_2$. We denote $Res(J)$ as the real part of the eigenvalues associated with matrix $J$, besides $\lambda_{\min}(S)$ and $\lambda_{\max}(S)$ as the smallest and the largest eigenvalues of $S$,respectively. For a function $\Phi =(\Phi_1,\Phi_2,...\Phi_n)^T \in \mathbb{R}^n$,$x=(x_1,x_2,...x_m)^T\in \mathbb{R}^m$, let $\frac{\partial \Phi}{\partial x^T} = (\frac{\partial \Phi_i(x)}{\partial x_j})_{ij}$. Matrix $A<B$ means $A-B$ is negative define matrix. 
\subsection{Definition and Lemma}
\label{sec:definition and lemma}
%We introduce some commonly used lemmas as
\newtheorem{definition}{Definition}
\begin{definition}
	\label{def:def1}
	(Robust stability\cite{yoszawa1962stability}) For perturbed nonlinear system $\dot{y} = f(t,y) + g(t,y)$, $y \in \mathbb{R}^n$,$f(t,0)=0$, $f,g \in C[I\times S_H,\mathbb{R}^n]$, $S_H=\{x|||x||\leq H\}$, if $\forall \varepsilon >0$, $\exists \delta_1(\varepsilon) >0$ and $\delta_2(\varepsilon) >0$ to make $||g(t,y)||\leq \delta_1(\varepsilon)$, $||y(0)||\leq \delta_2(\varepsilon)$ such that $||y(t)||<\varepsilon$, then the trivial solution of $\dot{y} = f(t,y) + g(t,y),f(t,0)=0$ exhibits robust stability.
\end{definition}
The concept of robust stability is defined in a way that it describes a property. Specifically, for any perturbed system, through appropriately adjusting the restricted range of the perturbation \(g(t,y)\) and the initial values of the trivial solution \(y(0)\), the perturbed solution can be maintained within an arbitrarily small interval. 

\newtheorem{lem}{Lemma}
\begin{lem}
	\label{thm:lem}
	If there exists Lyapunov function $V(e)=e^TPe$, $P=P^T>0$ for automous system $\dot{e}=f(e)$ to have $\dot{V}(e)\leq -\alpha V(e)$, $\alpha >0$, then $e$ will be exponentially converge to 0 and we define its convergency rate is $\alpha$.
\end{lem}
\begin{proof}
	Let $Z(t) = Z(e(t))=\dot{V}(e(t)) + \alpha V(e(t))\leq 0$, then
	\begin{equation}
		\begin{split}
			V(e(t))&=V(e(0))\exp(-\alpha t) + \int_0^t Z(r)\exp(-\alpha (t-r))dr\\
			&\leq V(e(0))\exp(-\alpha t) 
		\end{split}
	\end{equation}
	It is obvious that  
	\begin{equation}
			||e(t)||_2 
			\leq \sqrt{\frac{V(e(t))}{\lambda_{\min}(P)}}	
			\leq \sqrt{\frac{V(e(0))}{\lambda_{\min}(P)}} \exp(-\frac{\alpha}{2} t)	
	\end{equation}
	 The proof is completed.
\end{proof}

%\newtheorem{lem0}[lem]{Lemma}
%\begin{lem0}
%	\label{thm:lem0}
%	For a nonlinear system $\dot{x}(t) =f(x(t))$ with $f(0)=0$, $x\in \Omega \subset \mathbb{R}^n$, suppose there is a positive definite function $V(x)$ for any nonzero $x$ satisfying
%	\begin{equation}
%		\dot{V}(x) + \alpha V(x) + \beta V^{\gamma}(x) \leq 0
%	\end{equation}
%	where $\alpha >0$, $\beta >0$ and $0<\gamma < 1$. Then, the origin of system is fast finite-time stable in $\Omega$, and the settling time, depending on the initial state $x(0)=x_0$, is given by\cite{khoo2014multi}
%	\begin{equation}
%		T(x_0)\leq \frac{1}{\alpha (1-\gamma)}\ln (1+\frac{\alpha V^{1-\gamma}(x_0)}{\beta})
%	\end{equation}
%\end{lem0}

\newtheorem{lem2}[lem]{Lemma}
\begin{lem2}
	\label{thm:lem2}
	(The Bounded Real Lemma\cite{doi:10.1049/ip-cta:19982048}) Let $G$ be a linear time - invariant state-space model 
	\begin{equation}
		G:\begin{cases}
			\dot{x} = Ax + Bu\\
			y = Cx + Du	
		\end{cases}
	\end{equation}
	where $x\in \mathbb{R}^{n}$, $u\in \mathbb{R}^{m}$, $A\in \mathbb{R}^{n\times n}$,$B\in \mathbb{R}^{n\times m}$,$C\in \mathbb{R}^{p\times n}$, and $D \in \mathbb{R}^{p \times m}$. Then the following are equivalent.
	\begin{itemize}
		\item $||G||_{H_{\infty}}\leq \gamma$, $\gamma >0$.
		 \item there exists $X=X^T>0$ such that
		 \begin{equation}
		 	\begin{bmatrix}
		 		X A+A^T X & XB & C^T\\
		 		B^T X & -\gamma I & D^T \\
		 		C & D & -\gamma I \\
		 	\end{bmatrix} <0
		 \end{equation}.
	\end{itemize}
\end{lem2}

\newtheorem{lem1}[lem]{Lemma}
\begin{lem1}
	\label{thm:lem1}
		For a nonlinear system $\dot{x}(t) =f(x(t))$ with initial state $x(0)=x_0,f(0)=0$, $x\in \mathbb{R}^n$, suppose there is a positive definite function $V(x)$ for any nonzero $x$ satisfying
		\begin{equation}
			\label{equ:lem1_equ}
			\dot{V}(x) + \alpha V(x) - \beta V^{\frac{1}{2}}(x)  \leq 0
		\end{equation}
		where $\alpha >0$, $\beta >0$. Then for any state $x(t)$ there exists
		\begin{equation}
			\label{equ:lem2_equ}
			V^{\frac{1}{2}}(x(t)) \leq \frac{\beta}{\alpha} + (V^{\frac{1}{2}}(x_0) - \frac{\beta}{\alpha})e^{-\frac{\alpha}{2}t}
		\end{equation}
\end{lem1}

\begin{proof}
	We define \(s(t) = V^{\frac{1}{2}}(x)\), and then Eq.(\ref{equ:lem1_equ}) can be transformed into
	\begin{equation}
		\dot{s}(t) + \frac{\alpha}{2}s(t) \leq \frac{\beta}{2}
	\end{equation}
	The above equation can be further transformed into 
	\begin{equation}
		e^{\frac{\alpha}{2}t}\dot{s}(t) + \frac{\alpha}{2}e^{\frac{\alpha}{2}t}s(t) \leq \frac{\beta}{2}e^{\frac{\alpha}{2}t}
	\end{equation}
	which means
	\begin{equation}
		\frac{d}{dt}(e^{\frac{\alpha}{2}t}s(t)) \leq \frac{\beta}{2} e^{\frac{\alpha}{2}t}
	\end{equation}
	Integrating both sides of the inequality in the above equation from $0$ to $t$ simultaneously, we can obtain
	\begin{equation}
		e^{\frac{\alpha}{2}t}s(t) - s(0) \leq \frac{\beta}{\alpha} (e^{\frac{\alpha}{2}t} - 1)
	\end{equation}
	Therefore, the above equation can be further derived to obtain Eq.(\ref{equ:lem2_equ}). Thus, the proof is completed. 
\end{proof}

\subsection{High-dimensional PID controller}
Consider the following class of autonomous MIMO non-affine uncertain nonlinear systems within continuous and first-order differentiable $f \in \mathbb{R}^n$ and first-order differentiable disturbance $d \in \mathbb{R}^n$:
\begin{equation}
	\label{equ:class_equation}
	\dot{x} = f(x,u) + d
\end{equation}
where $x\in \mathbb{R}^n$ that an be measured by state estimator such as extended kalman filter, $u\in \mathbb{R}^m$ is the control input. Our control objective is to robust stabilize the above system and to make the controlled variable $x(t)$ converge to a desired reference value $x_r\in \mathbb{R}^n$ as much as possible, where signal tracking error $e=x_r - x$, for all initial states under the uncertainty first-order differentiable bounded disturbance as
\begin{equation}
	||\dot{d}||_2\leq L_{\dot{d}}
\end{equation}
We adopts multi-channel coupling parameters PID controller as
\begin{equation}
	\label{equ:control_strategy}
	\dot{u} = K_ie + K_p\dot{e} + K_d\ddot{e}
\end{equation}
and tuned PID controller parameters $K_p,K_i,K_d\in \mathbb{R}^{m\times n}$. 
Here, the control commands and its first-order derivative quantity are constrained as
\begin{equation}
	\label{equ:commands_constrain}
	u_{\min}\leq u \leq u_{\max},  \dot{u}_{\min} \leq \dot{u} \leq \dot{u}_{\max}
\end{equation}
where %$Q_1,Q_2 \in \mathbb{R}^{m\times m}$,
$u_{\min},u_{\max},\dot{u}_{\min},\dot{u}_{\max},\in \mathbb{R}^{m}$. 
We expect to exponentially stablize the Eq.(\ref{equ:class_equation}) within maximum stablization speed to resist the emergency situation. 

\subsection{The velocity representation of control system}
The use of the velocity representation of the nonlinear autonomous model Eq.(\ref{equ:class_equation}) makes it possible to obtain a model in the form\cite{bib:Leith}
\begin{equation}
	\label{equ:velocity_representation}
	\ddot{x} = \frac{\partial f}{\partial x^T}\dot{x} + \frac{\partial f}{\partial u^T}\dot{u}	+\dot{d}
\end{equation}
where $\frac{\partial f}{\partial x^T}$ and $\frac{\partial f}{\partial u^T}$ are the Jacobian of $f(x,u)$ respectively with respect to $x$ and $u$. Assuming the full-state reference signal is $x_r$, related second-order following error dynamic will be in the form:
\begin{equation}
	\label{equ:error_equation}
	\ddot{e} = \frac{\partial f}{\partial x^T}\dot{e} - \frac{\partial f}{\partial u^T}\dot{u} + d_e\\
\end{equation}
where 
\begin{equation}
	d_e = \ddot{x}_r - \frac{\partial f}{\partial x^T}\dot{x}_r - \dot{d}
\end{equation}
As the full-state high-dimension feedback PID control strategy Eq.(\ref{equ:control_strategy}) is applied to Eq.(\ref{equ:error_equation}), an error linear time-varying(LTV) state-space dynamic representation with perturbation can be given as follows
\begin{equation}
	\label{equ:error_strategy}	
	\begin{split}
		\begin{bmatrix}
			\ddot{e} \\ \dot{e}
		\end{bmatrix}
		&=
		\begin{bmatrix}
			I(K_d)^{-1}(\frac{\partial f}{\partial x^T} - \frac{\partial f}{\partial u^T}K_p) & -I(K_d)^{-1}\frac{\partial f}{\partial u^T}K_i \\
			I & O
		\end{bmatrix}
		\begin{bmatrix}
			\dot{e} \\ e
		\end{bmatrix}\\
		&+ 
		\begin{bmatrix}
			I(K_d)^{-1} d_e \\ 0
		\end{bmatrix}
	\end{split}
\end{equation}
where $I(K_d) = I+\frac{\partial f}{\partial u^T}K_d$, denote $\tilde{e}=(\dot{e}^T,e^T)^T$, $\tilde{d}=((I(K_d)^{-1}d)^T,0^T)^T$, we define the $K$ as
\begin{equation}
	K=(K_p,K_i)
\end{equation}
Thus Eq.(\ref{equ:error_strategy}) can be organized into
\begin{equation}
	\label{equ:simple_J}
	\dot{\tilde{e}} = \tilde{J}_K\tilde{e} + \tilde{d} \\
\end{equation}
where 
\begin{equation}
	\label{equ:tilde_J_K}
	\begin{split}
		\tilde{J}_K &= L_1 + L_2K\\
		L_1 = \begin{bmatrix}
			I(K_d)^{-1}\frac{\partial f}{\partial x^T} & O \\
			I & O
		\end{bmatrix} &,
		L_2 = \begin{bmatrix}
			-I(K_d)^{-1}\frac{\partial f}{\partial u^T} \\ O
		\end{bmatrix}
	\end{split}
\end{equation}
and the upper bound of $\tilde{d}$ is given by $L_{\tilde{d}}$, defined as
\begin{equation}
	L_{\tilde{d}}\leq||I(K_d)^{-1}d_e||_2\leq ||I(K_d)^{-1}||_2(L_{\dot{d}} + ||\ddot{x}_r - \frac{\partial f}{\partial x^T}\dot{x}_r||_2)
\end{equation}

\section{Dynamic Compensation Strategy of Controller}
\label{sec:Dynamic Compensation Strategy of Controller}
In this section, we aim to study the following question, how to design an appropriate parameter compensation \(\Delta K\), to achieve robust stabilization at the current point, after the parameter \(K\) of the high-dimensional PID controller render the perturbed system Eq.(\ref{equ:simple_J}) asymptotically stable at the origin, thus minimizing the range of the convergence domain under perturbed condition. The following theorem will lay the foundation for achieving this objective. 
\newtheorem{thm1}{Theorem}
\begin{thm1}
	\label{thm:thm1}
	For the time-varying perturbed system \(\dot{e}(t)=(J_K + \Delta J_K)e(t) + d\), $e(0)=e_0$, where \(||d||_2 \leq L_d\) and its compensation item is \(\Delta J_K\). If for a given \(Q = Q^T>0\), there exists \(P = P^T>0\) such that 
	\begin{equation}
		J_K^T P+P J_K\leq -Q
	\end{equation}
	and there exists \(\tau > 0\) for \(\Delta J_K\) such that
	\begin{equation}
		\label{thm:thm1_condition}
		\Delta J_K^TP+P\Delta J_K \leq Q-\tau I
	\end{equation}
	 Then there exists
	 \begin{equation}
	 	\label{thm:thm1_equ1}
	 	\sqrt{e^T(t)Pe(t)}\leq \eta(\tau) + (\sqrt{e^T_0Pe_0}-\eta(\tau))e^{-\frac{\tau }{2\lambda_{\max}(P)}t}
	 \end{equation}
	 where
	 \begin{equation}
	 	\eta(\tau) = \frac{2L_d\lambda_{\max}^2(P)}{\tau \lambda_{\min}^{\frac{1}{2}}(P)}
	 \end{equation}
	 Moreover, $\sqrt{e^T(t)Pe(t)}$ decreases monotonically as time $t$ progresses when $\sqrt{e^T(t)Pe(t)}>\eta(\tau)$.
\end{thm1}

\begin{proof}
	We construct the Lyapunov function \(V(e)=e^TPe,P=P^T>0\), and its first-order derivative is
	\begin{equation}
		\label{equ:thm1_pf1}
		\begin{split}
			\dot{V}(e) &= \dot{e}^TPe + e^TP\dot{e} \\
			& = [(J_K+\Delta J_K)e + d]^TPe + e^TP[(J_K+\Delta J_K)e + d] \\
			& = e^T(J_K^TP+PJ_K)e + e^T(\Delta J_K^TP+P\Delta J_K)e \\
			& + d^TPe + e^TPd\\
			& \leq -e^TQe + e^T(\Delta J_K^TP+P\Delta J_K)e + 2||e^TPd||_2 \\
			& \leq e^T(-Q+\Delta J_K^TP+P\Delta J_K)e + 2L_d\lambda_{\max}(P) ||e||_2\\
		%	& \leq 	 - \frac{\lambda_{\max}(Q-\Delta J_K^TP-P\Delta J_K)}{\lambda_{\max}(P)} V(e) + / \frac{2L_d\lambda_{\max}(P)}{\sqrt{\lambda_{\min}(P)}}V^{\frac{1}{2}}(e)
			& \leq e^T(-Q+\Delta J_K^TP+P\Delta J_K)e +
			\frac{2L_d\lambda_{\max}(P)}{\sqrt{\lambda_{\min}(P)}}V^{\frac{1}{2}}(e) \\
			& \leq - \frac{\tau}{\lambda_{\max}(P)} V(e) +\frac{2L_d\lambda_{\max}(P)}{\sqrt{\lambda_{\min}(P)}}V^{\frac{1}{2}}(e) \\
		\end{split}
	\end{equation}
	According to Lemma \ref{thm:lem1}, Eq.(\ref{thm:thm1_equ1}) holds. Moreover, when $\sqrt{e^T(t)Pe(t)}>\eta(\tau)$, Eq.(\ref{equ:thm1_pf1}) can be further transformed into $\dot{V}(e)\leq 0$, which simultaneously indicates that $V(e)$ decreases monotonically with time. In other words, $\sqrt{e^T(t)Pe(t)}$ also decreases monotonically with time. The proof is completed.  
\end{proof}
\begin{figure}[hbt!]
	\centering
	\includegraphics[width=\columnwidth]{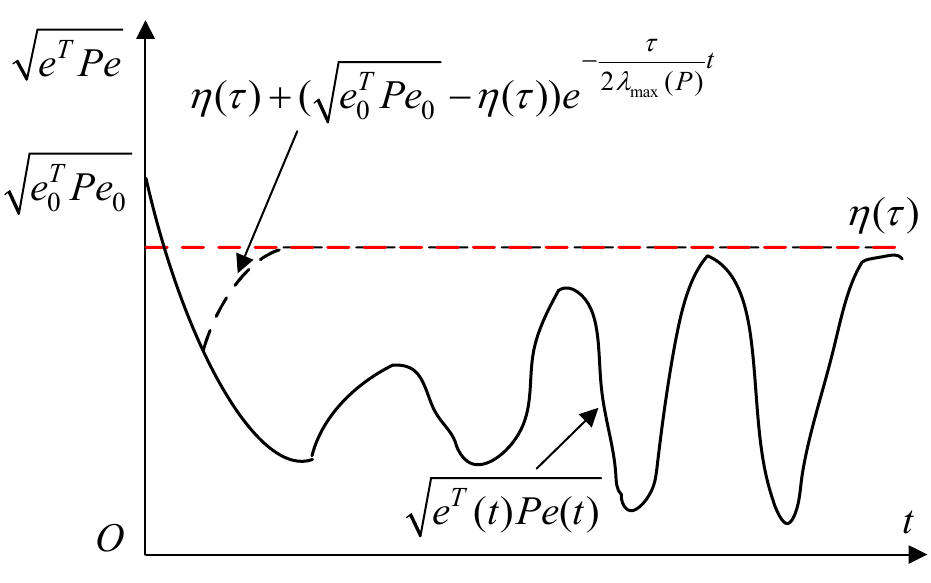}
	\caption{The convergence process of the elliptical form $\sqrt{e^T(t)Pe(t)}$ of error $e$, with respect to time $t$.}
	\label{fig:1}
\end{figure}
Theorem \ref{thm:thm1} elucidates the following principle. As depicted in Fig. \ref{fig:1}, when $K$ and $P$ are specified, for the perturbed system $\dot{e}(t)=(J_K+\Delta J_K)e(t)+d$, through the configuration of the compensation $\Delta J_K^T$ satisfying Eq. (\ref{thm:thm1_condition}), regardless of the initial error value $e_0$, the error $e(t)$ at any subsequent time $t$ can be constrained within a relatively small elliptical region $J(\tau)=\left\{e:\sqrt{e^TPe}\leq\eta(\tau)\right\}$. In light of Definition \ref{def:def1}, the perturbed system demonstrates robust stability within the global domain of attraction. 

\subsection{Optimal Configuration of Parameter \(K\)}
In Eq.(\ref{equ:tilde_J_K}), the value of \(\tilde{J}_K\) at \(\tilde{e} = 0\), namely \(\tilde{J}_K(0)\), is configured using the following approach to ensure the achievement of asymptotic stability at the origin. First, for the given \(Q = Q^T=\varepsilon_Q I>0\) and \(P = P^T=\varepsilon_P I>0\), the stability condition \(\tilde{J}_K(0)^T P + P \tilde{J}_K(0)=-Q\) can be transformed into:
\begin{equation}
	\label{equ:orgin_stable}
	L_1(0) + L_1^T(0) + L_2(0)K+K^TL_2(0)^T+\frac{\varepsilon_Q}{\varepsilon_P} I \leq O
\end{equation}
where \(L_1(0)\) and \(L_2(0)\) denote the values of \(L_1\) and \(L_2\) respectively at the equilibrium point \(\tilde{e} = 0\). Moreover, the convergence rate near the origin is determined by \(\frac{\varepsilon_Q}{\varepsilon_P}\), and the following theorem will illustrate this point. 
\newtheorem{thm2}[thm1]{Theorem}
\begin{thm2}
	For the linear time-invariant(LTI) unperturbed system \(\dot{\tilde{e}}=\tilde{J}(0)\tilde{e}\), if Eq.(\ref{equ:orgin_stable}) is satisfied, the \(\tilde{e}\) will exponentially converge to \(0\) at a convergence rate of \(\frac{\varepsilon_Q}{\varepsilon_P}\). 
\end{thm2}
\begin{proof}
	Eq.(\ref{equ:orgin_stable}) is equivalent to \(\varepsilon_P\tilde{J}_K(0)+\varepsilon_P\tilde{J}_K(0)^T+\varepsilon_Q I\leq0\). Herein, we construct the Lyapunov function \(V(\tilde{e}) = \varepsilon_P\tilde{e}^T\tilde{e}\), and it can be carried out by taking the derivative that $\dot{V}(\tilde{e}) \leq -\frac{\varepsilon_Q}{\varepsilon_P} V(\tilde{e})$. According to Lemma \ref{thm:lem}, the proof is completed. 
\end{proof}
In order to maximize the convergence rate of the system at the origin under unperturbed conditions, Eq.(\ref{equ:orgin_stable}) can be converted into the following eigenvalue problem (EVP) for solution, so as to obtain the most appropriate parameter \(K\).
\begin{equation}
	\label{equ:EVP1}
	\centering
	\begin{split}
		\min_{K} \lambda <0 \\
		L_K = L_1(0) + L_1^T(0) + L_2(0)K+K^TL_2(0)^T \leq \lambda I \\
	\end{split}
\end{equation} 
Furthermore, on the basis of above, in order to ensure the robustness of the system near the origin, the $H_{\infty}$ norm of the system can be restricted as minimum as possible. According to Lemma \ref{thm:lem2}, this problem is further transformed into the following generalized EVP.
\begin{equation}
	\label{equ:EVP1_2}
	\centering
	\begin{split}
		\min_{K} \lambda <0 \\
		\begin{bmatrix}
			L_K & O & \frac{1}{\varepsilon_P} I\\
			* & -\lambda I & O \\
			* & * & -\lambda I \\
		\end{bmatrix} <0 \\
	\end{split}
\end{equation} 

\subsection{Optimal Configuration of Compensation item \(\Delta K\)}
Having already determined the value of \(K\), a crucial problem addressed in this paper is how to optimally configure the current compensation \(\Delta K\). The goal is that the updated parameter \(K'=K + \Delta K\) enables the system described by Eq.(\ref{equ:simple_J}) to achieve fast robust stability under perturbed condition and ensures the minimum range of domain of attraction $J(\tau)$ possible. Compared with the matrix \(\tilde{J}_K(0)\) at the origin, the increment \(\Delta\tilde{J}_{K+\Delta K}(\tilde{e})\) of the compensated matrix \(\tilde{J}_{K + \Delta K}\) at the present error \(\tilde{e}\) can be characterized as follows: 
\begin{equation}
	\begin{split}
		\Delta\tilde{J}_{K+\Delta K} &= \tilde{J}_{K + \Delta K}(\tilde{e}) - \tilde{J}_K(0) \\
		&= \Delta L_1 + \Delta L_2 K + L_2(\tilde{e}) \Delta K
	\end{split}
\end{equation}
where $\Delta L_1 = L_1(\tilde{e}) - L_1(0)$, $\Delta L_2 = L_2(\tilde{e}) - L_2(0)$. Similarly, we set \(Q = Q^T=\varepsilon_Q I>0\) and \(P = P^T=\varepsilon_P I>0\), then there exists the following theorem, which reveals how to configure the compensation item to enable the error \(\tilde{e}\) to achieve fast robust stability.
\newtheorem{thm3}[thm1]{Theorem}
\begin{thm3}
	For the linear time-variant(LTV) perturbed system \(\dot{\tilde{e}}(t)=(J_K + \Delta J_K)\tilde{e}(t) + \tilde{d}\), $\tilde{e}(0)=\tilde{e}_0$, $||\tilde{d}||_2\leq L_{\tilde{d}}$, if Eq.(\ref{equ:orgin_stable}) is satisfied for given $\varepsilon_P>0$, $\varepsilon_Q>0$, and there exists $\tau >0$ such that
	\begin{equation}
		\label{thm:thm3_condition}
		\begin{split}
			L_2(\tilde{e})\Delta K + \Delta K^T L_2^T(\tilde{e}) + \Delta L_2 K + K^T \Delta L_2^T \\
			+ \Delta L_1 + \Delta L_1^T + \frac{\tau - \varepsilon_Q}{\varepsilon_P} I \leq O 
		\end{split}
	\end{equation} 
	Then the perturbed system exhibits robust stability, and $\tilde{e}$ will converge to the invariant set
	\begin{equation}
		\label{equ:invariant set}
		J(\tau)=\left\{e:\sqrt{e^TPe}\leq\eta(\tau) = \frac{2L_{\tilde{d}}\varepsilon_P^{\frac{3}{2}}}{\tau} \right\}
	\end{equation} 
\end{thm3}
\begin{proof}
	Eq. (\ref{thm:thm3_condition}) is equivalent to 
	\begin{equation}
		\Delta \tilde{J}_{K + \Delta K}^TP+P\Delta \tilde{J}_{K + \Delta K} \leq Q-\tau I
	\end{equation}
	where $P = \varepsilon_P I$ and $Q=\varepsilon_Q I$. According to Theorem \ref{thm:thm1}, it can be concluded that the conclusion holds. The proof is completed. 
\end{proof}
In order to minimize the upper bound of the $J(\tau)$ as shown in Eq.(\ref{equ:invariant set}) as much as possible, while fully ensuring the maximization of \(\frac{\varepsilon_Q}{\varepsilon_P}\), it is necessary to simultaneously ensure the maximization of \(\tau\). This objective can also be achieved by optimizing the compensation amount \(\Delta K\) given an appropriate \(K\). Similarly, this problem can be converted into the following EVP.
\begin{equation}
	\label{equ:EVP2}
	\centering
	\begin{split}
		&\min_{\Delta K} \lambda <0 \\
		& \varepsilon_PL_2(\tilde{e})\Delta K + \varepsilon_P\Delta K^T L_2^T(\tilde{e}) + \varepsilon_P\Delta L_2 K + \varepsilon_PK^T\Delta L_2^T \\
	    & + \varepsilon_P\Delta L_1 + \varepsilon_P\Delta L_1^T - \varepsilon_Q I \leq \lambda I
	\end{split}
\end{equation} 
Through the foregoing analysis, we have successfully transformed the online tuning problem of high-dimensional PID controller parameters into the solution of two classes of adjacent eigenvalue problems (EVPs). Initially, to maximize the convergence rate of the controller at the equilibrium point, we determine the optimal \(K\). Subsequently, having obtained \(K\), we maximize the finite-time convergence rate at the current state error to obtain the dynamic compensation item \(\Delta K\) for the parameter \(K\). The adjusted value \(K + \Delta K\) is then applied to the high-dimensional PID controller, which further optimizes the convergence rate, ensuring finite-time stability of the controller in the case where the system is perturbed.

\section{Simulation}
\label{sec:Simulation}
\subsection{Experimental configuration}
In this section, we establish the theoretical validation experiment framework explored in our study. For the sake of simplicity, we adopt the kinematic model of a fixed - wing aircraft \cite{bib:Samir} in the ground coordinate system along the $\gamma$ and $\chi$ directions, subject to perturbations $d_{\chi}$ and $d_{\gamma}$, as the nonlinear controlled model. This model is described by the following equations:
\begin{equation}
	\label{equ:perturbations_model}
	\begin{cases}
		\dot{\chi} = g\tan\phi/V + d_{\chi}\\
		\dot{\gamma} = g(n_z\cos\phi - \cos\gamma)/V + d_{\gamma} 
	\end{cases}
\end{equation}
Let the state vector be denoted as $x = (\chi,\gamma)^T$, the input commands as $u = (\phi,n_z)^T$, the state differential quantity without perturbations as $f = (\dot{\chi} - d_{\chi},\dot{\gamma} - d_{\gamma})^T$. The corresponding Jacobians are as follows:
\begin{equation}
	\label{equ:states_differential}
	\frac{\partial f}{\partial x^T} = \begin{bmatrix}
		0 & 0\\ 0 &  g\sin\gamma/V
	\end{bmatrix},
	\frac{\partial f}{\partial u^T} = \begin{bmatrix}
		g\sec^2\phi/V & 0\\
		-gn_z\sin\phi/V & g\cos\phi/V
	\end{bmatrix}
\end{equation}
Our aim is to configure appropriate high - dimensional PID controller input commands $\phi$ and $n_z$ to achieve the ideal tracking performance of the reference signals $\gamma_c$ and $\chi_c$, assuming that the velocity $V$ is well - maintained. Specifically, we emphasize the stabilization process of the reference errors $e_{\gamma}=\gamma_c - \gamma$, $e_{\chi}=\chi_c - \chi$ and their first - order differential quantities $\dot{e}_{\gamma}$ and $\dot{e}_{\chi}$ under perturbations, which are bounded by the upper - bound constants $L_{d_{\chi}}$ and $L_{d_{\gamma}}$. Here, the perturbation $d$ is modeled as uniform white noise, expressed in the following form:
\begin{equation}
	\label{equ:d}
	d \sim U(-\frac{L_d}{2}\mathbf{1}_n,\frac{L_d}{2}\mathbf{1}_n)
\end{equation}
The critical hyperparameters for this experiment are presented in TABLE \ref{tab:hyperparameters}. It is specified that the initial derivatives of all states are zero. 
\begin{table}[htbp]
	\caption{Hyperparameter declarations}
	\label{tab:hyperparameters}
	\centering
	\begin{tabular}{cccc}
		\hline
		\bf{Declaration} & \bf{Hyperparameter} & \bf{Value} & \bf{Unit}\\ 
		\hline % 横线
		Simulation timespan & $T$ & [0,20] & $s$\\
		Acceleration of gravity & $g$ & 9.81 & $m/s^2$ \\
		Consolidated velocity & $V$ & 25 & $m/s$ \\
		Initial climb angle & $\gamma$ & $\pi/4$& $rad$\\
		Initial azimuth angle & $\chi$ & $\pi/3$& $rad$\\
		Initial roll angle & $\phi$ & $\pi/3$& $rad$\\
		Initial overload & $n_z$ & $1$& $-$\\
		Reference climb angle & $\gamma_c$ & $0$& $rad$\\
		Reference azimuth angle & $\chi_c$ & $0$& $rad$\\
		Reference roll angle & $\phi_c$ & $0$& $rad$\\
		Reference overload & $n_{zc}$ & $0$& $-$\\
		Upper-bound constant of $d_{\chi}$ & $L_{d_{\chi}}$ & $0.5$ & $rad$ \\
		Upper-bound constant of $d_{\gamma}$ & $L_{d_{\gamma}}$ & $0.5$ & $rad$ \\
		Controller parameter & $K_d$ & $O$ & $-$  \\
		\hline 
	\end{tabular}
\end{table}

\subsection{Result Analysis of adding Optimal Compensation}
In this section, we perform a comparative analysis of the disparities in the time - series curves of the state - variable errors \(e_{\gamma}\), \(e_{\chi}\) and the first - order differentials of these errors \(\dot{e}_{\gamma}\), \(\dot{e}_{\chi}\) subsequent to the addition of the compensation terms. Initially, for the present state conditions, we solve the EVP problem of Eq.(\ref{equ:EVP1}) to compute the original high - dimensional PID controller parameters \(K\) as
\begin{equation}
	K_p = \begin{bmatrix}
		1.0159 & 0 \\
		0 & 1.0159
	\end{bmatrix}
	,
	K_i = \begin{bmatrix}
		2.0406 & 0 \\
		0 & 2.0406
	\end{bmatrix}
\end{equation}  
The error curves of the states corresponding to the parameter \(K\) are presented as the green dashed lines in Fig. \ref{fig:2}. Evidently, while the time - series profiles of the error terms \(e_{\gamma}\) and \(\dot{e}_{\gamma}\) can effectively achieve a satisfactory level of stabilization, however, the error in the horizontal direction, \(e_{\chi}\), along with its first - order derivative \(\dot{e}_{\chi}\), exhibit pronounced oscillations. As a consequence, the stabilization performance for these components is rather poor. Meanwhile, we calculate the dynamic compensation terms $\Delta K_p$ and $\Delta K_i$ using the solution of EVP problem Eq.(\ref{equ:EVP2}) as follows
\begin{equation}
		\Delta K_p = \begin{bmatrix}
		-0.3809 & -0.7744 \\
		7.2946 & 4.1368
	\end{bmatrix}
	,
	\Delta K_i = \begin{bmatrix}
		-1.4029 & 0 \\
		1.1045 & 3.0609
	\end{bmatrix}
\end{equation}
\begin{figure}[htbp]
	\centering
	\begin{tabular}{cc}
		\includegraphics[width=0.48\columnwidth]{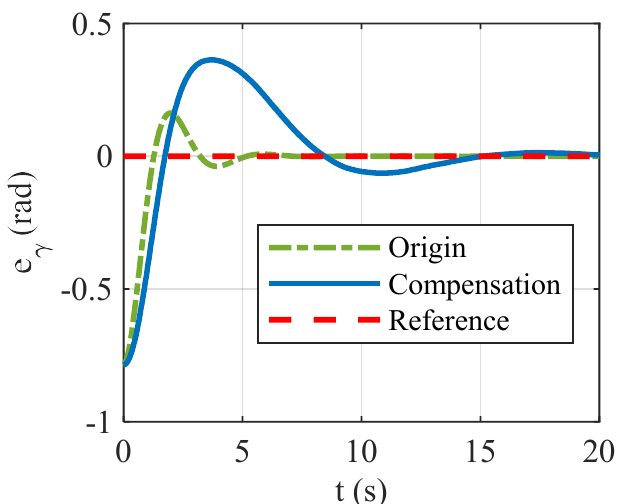} &
		\includegraphics[width=0.48\columnwidth]{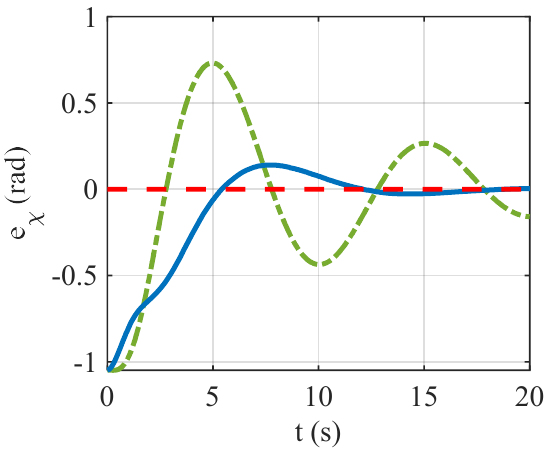} \\
		(a) & (b)\\
		\includegraphics[width=0.48\columnwidth]{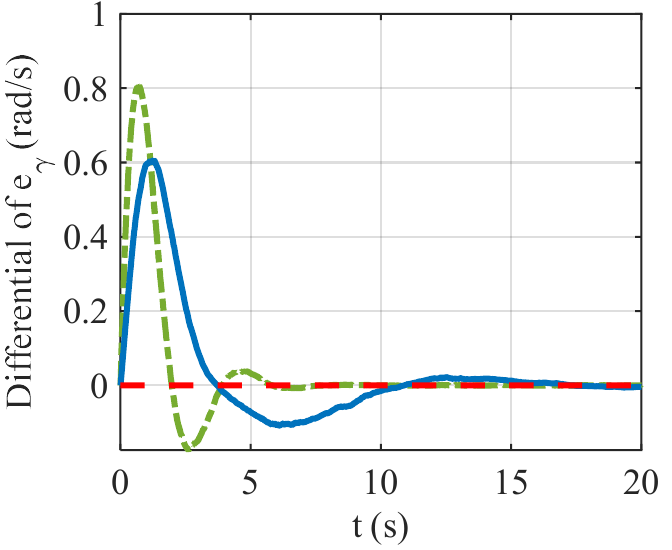} &
		\includegraphics[width=0.48\columnwidth]{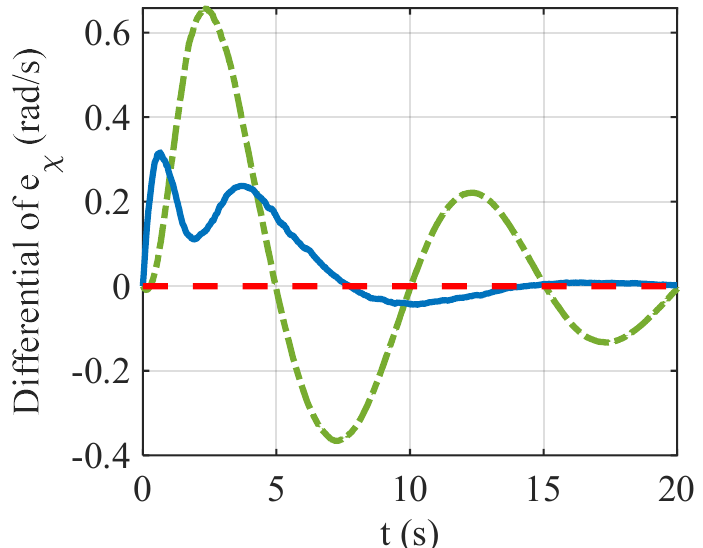} \\
		(c) & (b)
	\end{tabular}
	\caption{Comparison of the time-series profiles of the states $e_{\gamma}$, $e_{\chi}$ and their differentials $\dot{e}_{\gamma}$, $\dot{e}_{\chi}$ before and after adopting the original parameters $K$ and the parameters $K + \Delta K$ after optimal compensation. (a): Comparison of the $e_{\gamma}$; (b): Comparison of the $e_{\chi}$; (c): Comparison of the $\dot{e}_{\gamma}$; (d): Comparison of the $\dot{e}_{\chi}$.}
	\label{fig:2}
\end{figure}
 After applying these to obtain the new controller parameters $K_p+\Delta K_p$ and $K_i+\Delta K_i$ for compensation - based control, the resultant control curves are depicted as the blue solid lines in Fig. \ref{fig:2}. Evidently, while striving to maintain the error - stabilization of the original $e_{\gamma}$ and $\dot{e}_{\gamma}$, the pronounced oscillations of the error terms $e_{\chi}$ and $\dot{e}_{\chi}$ are effectively mitigated. This accomplishment leads to the efficient stabilization of all error - related states, thereby vividly demonstrating the robust stability inherent in the high - dimensional PID controller. 
\begin{figure}[htbp]
 	\centering
 	\includegraphics[width=\columnwidth]{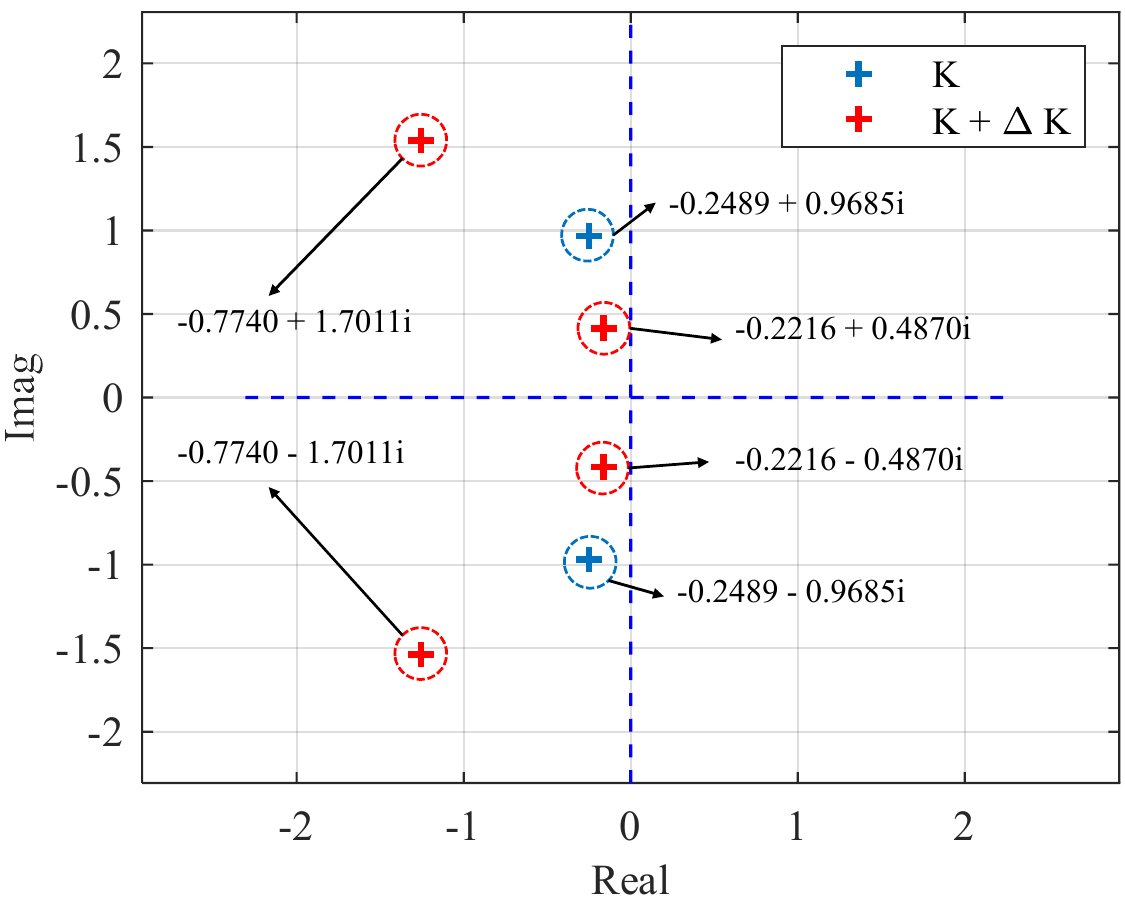}
 	\caption{A comparison is made of the eigenvalue distributions of the matrices before compensation $\tilde{J}_K(0)$ and the matrices after compensation $\tilde{J}_{K+\Delta K}(0)$ at the origin.}
 	\label{fig:3}
 \end{figure}
 Moreover, we carried out an in - depth analysis of the eigenvalue distributions of the state - transition matrices \(\tilde{J}_K\) and \(\tilde{J}_{K + \Delta K}\) at the origin both prior to and subsequent to compensation, as shown in Fig. \ref{fig:3}. The eigenvalue distributions serve as a crucial indicator, effectively reflecting the anti - disturbance capabilities of the state variables in the proximity of the origin. Notably, regardless of whether the compensation is applied or not, all eigenvalues are positioned within the left - hand half of the complex plane. This finding strongly implies that the system maintains certain asymptotic stability properties in the vicinity of the error zero - point even when subjected to external perturbations.
 Specifically, for \(\tilde{J}_K(0)\), it features a pair of conjugate multiple eigenvalues in the left - hand half - plane of the complex plane. In contrast, after compensation, the eigenvalues of \(\tilde{J}_{K+\Delta K}(0)\) display a more dispersed pattern within the left - hand half - plane of the complex plane and are situated at a greater distance from the imaginary axis. This distinct difference potentially accounts for the fact that the controller parameters post - compensation result in smaller oscillation amplitudes and more effective stabilization of the error at the origin.
 
\subsection{Performance Comparison involving Compensation}
We leverage the fundamental metrics of Integral Time Absolute Error (ITAE), Peak Time (PT), and Maximum Overshoot (MO) to quantitatively evaluate the performance of the high-dimensional PID controller under diverse conditions. The ITAE represents the average of the absolute differences between the actual signal and the reference signal, integrated over a specific time period. The PT refers to the time elapsed for a signal to ascend from a predefined initial value to its peak. The MO denotes the maximum extent by which a response surpasses its final value. Through these metrics, we can comprehensively assess the controller's effectiveness and adaptability in different scenarios. 
\begin{figure}[htbp]
	\centering
	\begin{tabular}{cc}
		\includegraphics[width=0.23\textwidth]{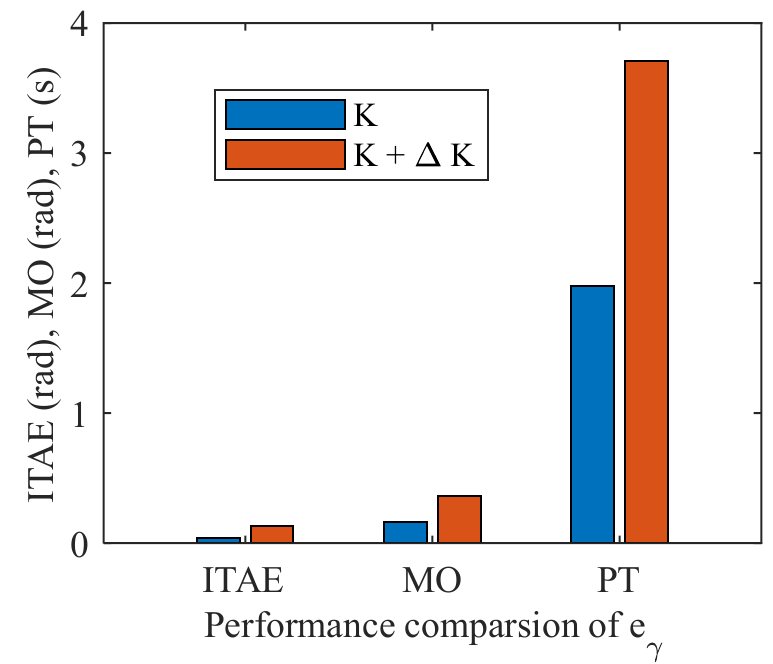} &
		\includegraphics[width=0.23\textwidth]{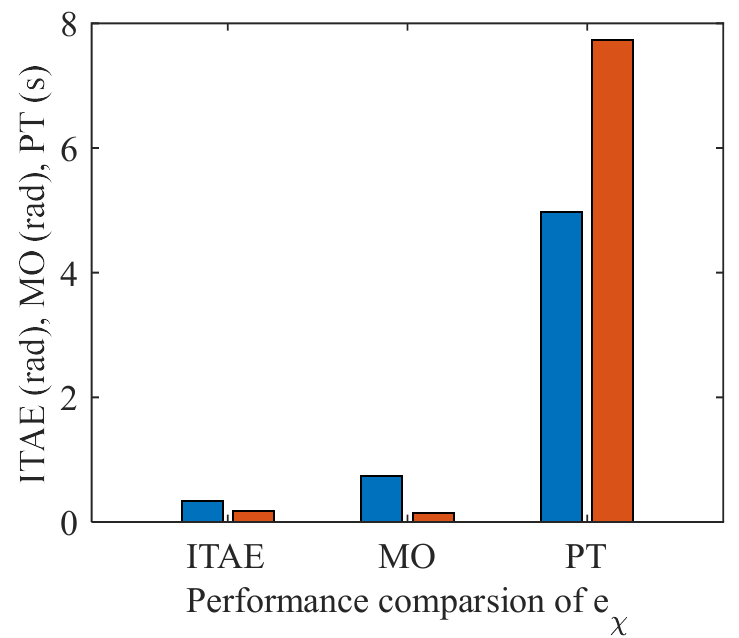} \\
		(a) & (b) \\
		\includegraphics[width=0.23\textwidth]{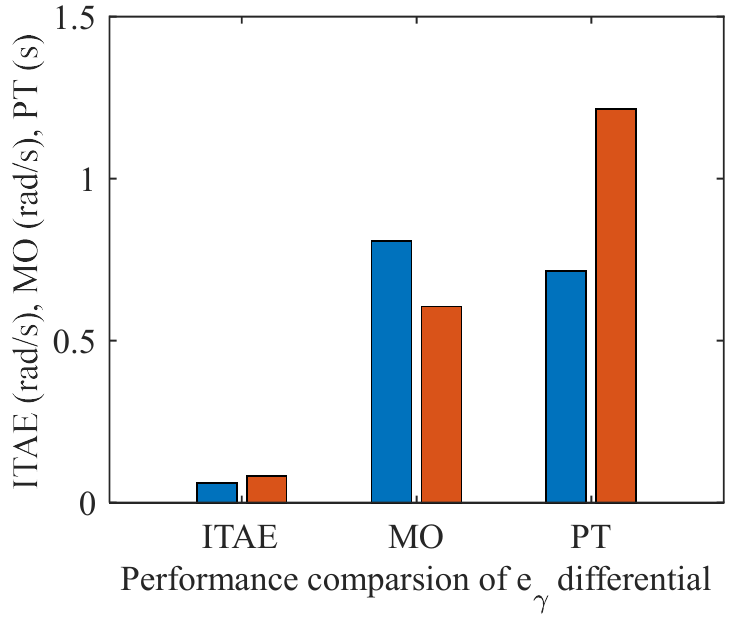} &
		\includegraphics[width=0.23\textwidth]{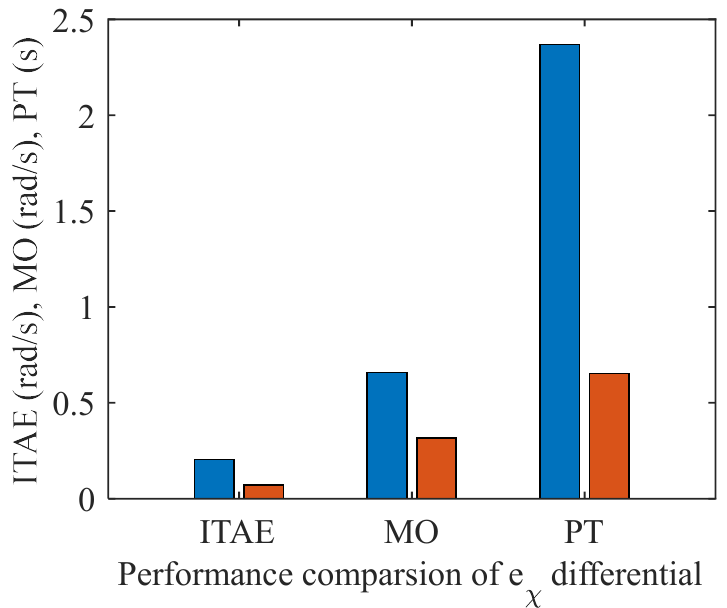} \\
		 (c) & (d) \\
	\end{tabular}
	\caption{Quantitative comparison of differences in high-dimensional PID controller before and after compensation using ITAE, PT, and MO indictors from the perspectives of the states $e_\gamma$, $e_\chi$, $\dot{e}_\gamma$, and $\dot{e}_\chi$. (a): $e_{\gamma}$; (b): $e_{\chi}$; (c): $\dot e_{{\gamma}}$; (d): $\dot e_{{\chi}}$.}
	\label{fig:4}
\end{figure}
We conducted a quantitative comparative analysis of numerous key indicators within the time - series curves presented in Fig. \ref{fig:2}, as depicted in Fig. \ref{fig:4}. Specifically, the high - dimensional PID controller exhibits superior comprehensive performance in the $e_{\chi}$ and $\dot{e}_{\chi}$ dimensions. It attains a lower ITAE, a reduced MO, and a smaller average value of PT. Notably, this enhanced performance in the $e_{\chi}$ and $\dot{e}_{\chi}$ dimensions comes at the expense of a certain degradation in $e_{\gamma}$ and $\dot{e}_{\gamma}$. Nevertheless, such a trade - off is acceptable, as clearly shown in Fig. \ref{fig:2}, even under this scenario, $e_{\gamma}$ and $\dot{e}_{\gamma}$ remain in a convergent state. Considering the key indicators ITAE and MO for both $e_{\gamma}$ and $e_{\chi}$, after compensating the controller parameters, the performance improvement in terms of $e_{\chi}$ is far more significant than the performance degradation in terms of $e_{\gamma}$. Based on the above - mentioned analysis, overall, the controller with parameter compensation indeed achieves a more optimized control effect. 

\section{Conclusion}
This research is dedicated to the online tuning problem of parameters for high - dimensional PID controller in general nonlinear perturbed system. Initiating from the velocity form of the error dynamic system, the error - stabilization problem is formulated as a robust stability problem of a class of linear time - varying perturbed linear system. Regarding this issue, the scope of the invariant set for error convergence under specific conditions is theoretically analyzed. Consequently, an appropriate compensation for the dynamic controller parameters can be devised to minimize the extent of this error - convergence domain. In this procedure, the configuration of controller parameters is decomposed into two phases: the design of exponential - stabilization parameters at the origin and the design of the optimal compensation at the current state. Moreover, it is concretely transformed into the corresponding two - stage eigenvalue problem (EVP) for expeditious solution. Experimental results verify that our online tuning approach can effectively eliminate the oscillation problem during the convergence of state errors, guarantee the asymptotic stability of the error dynamics system at the origin, and realize an enhancement in comprehensive performance, demonstrating excellent robust stability.

%%%%% ???? %%%%%%%%%%
\bibliographystyle{unsrt}
\bibliography{ref}

\end{document}